\documentclass[prx,reprint,floatfix,twocolumn,10pt,showpacs,superscriptaddress,nofootinbib]{revtex4-1}
\usepackage[T1]{fontenc}
\usepackage{amsmath}
\usepackage{amsthm}
\usepackage{graphicx}
\usepackage{times}
\usepackage{amssymb}
\usepackage{hyperref}
\hypersetup{
colorlinks=true,
linkcolor=red,
citecolor=red,}
\usepackage{textcomp}
\usepackage{xcolor}
\usepackage{enumerate}
\usepackage[mathscr]{euscript}
\usepackage{mathtools}
\usepackage{needspace}
\usepackage[mathscr]{euscript}
\usepackage{textcomp}
\usepackage{calrsfs}
\usepackage{bbm}
\usepackage{bm}
\usepackage[normalem]{ulem}
\usepackage{centernot}
\usepackage{stmaryrd}
\usepackage{longtable}

\newcommand{\T}{\mathscr{T}}
\newcommand{\PP}{\mathscr{P}}
\newcommand{\OO}{\mathcal{O}}
\newcommand{\E}{\mathcal{E}}
\newcommand{\I} {\mathcal I}
\newcommand{\M}{\mathcal{M}}
\newcommand{\G}{\mathcal{G}}
\newcommand{\Ident} {\hat{\mathbf 1}}

\newcommand{\s}[1]{_{\rm #1}}
\newcommand{\set}{\mathcal{S}}
\newcommand{\vs}[1]{{\boldsymbol{#1}}}
\newcommand{\conv}{{\operatorname{conv}}}

\newtheorem{definition}{Definition}

\newtheorem{lemma}{Lemma}
\newtheorem{theorem}{Theorem}
\newtheorem{crit}{Criterion}
\newtheorem{corol}{Corollary}

\newcommand{\ketbra}[1]{ | #1 \rangle\!\langle #1 |}

\newcommand{\ketbrax}[2]{ | #1 \rangle\s{#2}\langle #1 |}
\newcommand{\bra}[1] { \langle #1 | }
\newcommand{\ket}[1] { | #1 \rangle }

\begin{document}
\title{Quantum correlations and global coherence in distributed quantum computing}
 \author{Farid Shahandeh}
 \email{Electronic address: shahandeh.f@gmail.com}
  \affiliation{Centre for Quantum Computation and Communication Technology, School of Mathematics and Physics, University of Queensland, St Lucia, Queensland 4072, Australia}
   \affiliation{Department of Physics, Swansea University, Singleton Park, Swansea SA2 8PP, United Kingdom}
   \author{Austin P. Lund}
  \author{Timothy C. Ralph}
 \affiliation{Centre for Quantum Computation and Communication Technology, School of Mathematics and Physics, University of Queensland, St Lucia, Queensland 4072, Australia}

\begin{abstract}
Deviations from classical physics when distant quantum systems become correlated are interesting both fundamentally and operationally.
There exist situations where the correlations enable collaborative tasks that are impossible within the classical formalism.
Here, we consider the efficiency of quantum computation protocols compared to classical ones as a benchmark for separating quantum and classical resources and argue that the computational advantage of collaborative quantum protocols in the discrete variable domain implies the nonclassicality of correlations.
By analysing a toy model, it turns out that this argument implies the existence of quantum correlations distinct from entanglement and discord. 
We characterize such quantum correlations in terms of the net global coherence resources inherent within quantum states and show that entanglement and discord can be understood as special cases of our general framework.
Finally, we provide an operational interpretation of such correlations as those allowing two distant parties to increase their respective local quantum computational resources only using locally incoherent operations and classical communication.
\end{abstract}


\maketitle


The complete characterization of correlations between constituent elements of quantum systems is important both fundamentally and operationally.
Two well-known examples of such attempts are quantum entanglement and discord.
Entanglement is a resource for many nonlocal tasks~\cite{Ekert1991,Bennett1993} that cannot be created between spatially separated subsystems using local operations and classical communication (LOCC)~\cite{Werner1989}.
However, in many other tasks entanglement is known to play no or very minor role~\cite{DQC1,Braunstein1999,Lanyon2008}, putting forward quantum discord~\cite{Henderson2001,Ollivier2001,Modi2012,Misra2015} as a necessary resource~\cite{Horodecki2005,Datta2008,Madhok2013}, although there are ongoing controversies~\cite{Modi2012,Datta2012,Ferraro2012}.
The latter arises from the discrepancy between the entropic measures of correlations in classical and quantum physics
showing that not all the information encoded via LOCC into spatially separated systems can be extracted using the same type of operations~\cite{Bennett1999,Horodecki2005}.

Here we offer a new viewpoint on the quantumness of correlations.
Our ultimate objective is three-fold:
first, to understand the fundamental border (if any) between classical and quantum correlations in light of the nonclassical power of quantum computers;
second, to put forward a novel unified and consistent framework for characterizing quantum correlations in both continuous and discrete variable domains~\cite{Ferraro2012,Shahandeh2017NLBS};
third, to obtain a deeper understanding of the resources that might be responsible for the nonclassical power of quantum computation models.
The present manuscript focuses on the first aim and provides a proposal for the second one, in complement to our recent investigation of the continuous variable protocols~\cite{Shahandeh2017NLBS}.
In view of our results, we also obtain a new perspective on the third goal.
To this end, we first examine nonclassicality from two viewpoints, namely, that of computational science and the resource theory of coherence~\cite{Brandao2015,Baumgratz2014,Winter2016} and provide two nonclassicality criteria based on them.
We establish a fundamental correspondence between classical computation protocols and the formalism of quantum coherence.
This close affinity benchmarks computational efficiency for quantum-classical separation and gives rise to an equivalence between our nonclassicality criteria.
We then introduce a toy model called \emph{nonlocal deterministic quantum computing with two qubits} (NDQC2) which performs a nonlocal collaborative computation exponentially faster than any classical algorithm via correlation measurements without using any entanglement or discord.
Making use of our computational efficiency benchmark, this protocol shows quantumness of correlations that are not captured by the standard classification in quantum information theory.
This is complementary to a similar conclusion for the continuous variable domain where we showed that the nonlocal \textsc{BosonSampling} protocol contains correlations that cannot be efficiently simulated on a classical computer while the input and output mixed states contain no entanglement or discord~\cite{Shahandeh2017NLBS}.
In contrast to this, for the specific case of pure state quantum computations, entanglement is known to be required for a computational advantage over classical algorithms~\cite{Vidal2003}.
Inspired by this feature of NDQC2 and the role of coherence as a primitive property of quantum systems~\cite{Jiang2013,Vogel2014,Streltsov2015,Gholipour2016,Killoran2016,Ma2016,Tan2016} in quantum computation~\cite{Ma2016,Matera2016}, we show that the nonclassical advantage of the correlations within NDQC2 can be quantitatively explained in terms of the net global coherence inherent in the input and output states to the protocol.
We thus argue that the net global quantum-coherence should be understood as the more general concept of quantum correlations.
We show the relevance of our definition by proving that the current standard hierarchy of quantum correlations defines special classes of globally-coherent states and further providing an operational interpretation for such correlations.
To be specific, quantum correlations as presented here are manifested in the ability of two distant parties to increase their local quantum computational resources by applying classical operations locally and exploiting classical communication.


\section{Preliminaries}


\subsection{Standard Correlations in Quantum Information}\label{CQI}

From the viewpoint of quantum information theory, not all the global information can be encoded within or decoded from a bipartite (or multipartite) physical system via local operations and classical communication.
This leads to the following hierarchy of quantum-correlated states~\cite{Werner1989,Oppenheim2002,Horodecki2005}:
\begin{enumerate}[(i)]
\item entangled states that cannot be written in the \textit{separable} form $\hat{\varrho}_{\rm AB}{=}\sum_i p_i \hat{\varrho}_{{\rm A};i}{\otimes}\hat{\varrho}_{{\rm B};i}$;

\item two-way quantum correlated (discordant) states which cannot be represented via a set of locally orthogonal states on either side;

\item one-way nondiscordant (or one-way quantum-classical correlated) states that can be written as $\hat{\varrho}_{\rm AB}{=}\sum_{j} p_{j} \hat{\varrho}_{{\rm A};j}{\otimes}|j\rangle_{\rm B}\langle j|$, or $\hat{\varrho}_{\rm AB}{=}\sum_{i} p_{i} |i\rangle_{\rm A}\langle i|{\otimes}\hat{\varrho}_{{\rm B};i}$ using at most one set of locally orthonormal states;

\item two-way or fully nondiscordant (or strictly classical-classical) states admitting the form $\hat{\varrho}_{\rm AB}{=}\sum_{ij} p_{ij} |i\rangle_{\rm A}\langle i|{\otimes}|j\rangle_{\rm B}\langle j|$.

\end{enumerate}
The latter simply encode the joint probability distributions $\{p_{ij}\}$ using locally orthonormal states and have previously been assumed to possess no quantum advantage in a nonlocal information processing task.
The correlations within each of the classes above are usually measured using an entropic function as the discrepancy between their total-correlation contents and the amount accessible via LOCC, most commonly called quantum discord~\cite{Henderson2001,Ollivier2001}.
For a bipartite quantum state, quantum discord is asymmetric and a quantum state has zero discord both from Alice to Bob and vice versa if and only if it is of the form (iv).
For this reason, speaking of classical correlations is assumed to be synonymous with nondiscordant states.
In addition, every entangled state is necessarily discordant.
Quantum discord is thus considered as the most general measure of quantum correlations in quantum information theory~\cite{Oppenheim2002,Horodecki2005}.
In a recent work, however, we considered the computational advantage obtained in the nonlocal \textsc{BosonSampling} quantum computation protocol, that exploits particular types of mixed states, to show that there exists quantum correlations that are not captured by this picture, namely the global P-function nonclassicality~\cite{Shahandeh2017NLBS}.


\subsection{Quantum Coherence}

The resource theory of coherence comprises 
(i) a set of pure quantum states as extreme points $\mathcal{E}{=}\{|i\rangle\}$ which generates the set of (cost) free states of the theory as its convex hull $\set\s{inc}{=}{\rm conv}\{|i\rangle\langle i|{:}|i\rangle{\in}\mathcal{E}\}$;
(ii) the set of free transformations $\mathcal{O}\s{inc}$ which leave the set of free states invariant.
The extreme points, usually termed \emph{the computational basis}, in general, does not need to satisfy any orthogonalization or completeness conditions.
This is, for example, the case for the nonclassicality theory of continuous variable bosonic systems in which the set of bosonic coherent states $\{|\alpha\rangle{:}\alpha{\in}\mathbb{C}\}$, as extreme points, are nonorthogonal and overcomplete~\cite{GlauberBook}.
It is also clear that having infinite freedom, our choice of the computational basis depends on the physical system of interest and fundamental or operational restrictions.
For instance, in photonics the computational basis can be chosen to be the vertical and horizontal, or the diagonal and antidiagonal components of the radiation field.
In contrast, in an atomic realization of qubits, the preferred computational basis could be the energy eigenstates to which the systems decohere.

We assume that the computational basis is a finite complete orthonormal basis.
In particular, we are interested in collaborative global computations consisting of two local computations.
In such scenarios, when the local computational bases for the two parties, Alice and Bob, are $\mathcal{E}_{\rm A}{=}\{|i\rangle_{\rm A}\}$ and $\mathcal{E}_{\rm B }{=}\{|j\rangle_{\rm B}\}$, respectively, the \emph{global} computational basis is given by $\mathcal{E}_{\rm AB }{=}\mathcal{E}_{\rm A}{\otimes}\mathcal{E}_{\rm B}{=}\{|i\rangle_{\rm A}\otimes|j\rangle_{\rm B}\}$ if Alice and Bob are confined to separable operations (S).

Incoherent states are invariant under various classes of free operations.
Several of such operations for the resource theory of coherence have been studied so far, e.g.,\ general, strict~\cite{Winter2016}, and genuine incoherent operations~\cite{deVicente2017}.
A review of these operations and their operational meaning can be found in Refs.~\cite{Chitambar2016-2,Streltsov2017,Streltsov2017RMP}.
The first class of interest here, is called the {\it general} incoherent operations as the most general incoherent operations possible and possess Kraus decomposition $\Lambda(\cdot)=\sum_i \hat{F}_i(\cdot)\hat{F}_i^\dag$ such that $\sum_i \hat{F}_i^\dag\hat{F}_i=\Ident$ and $\hat{F}_i\set\s{inc}\hat{F}_i^\dag \subset \set\s{inc}$ for all $i$~\cite{Winter2016}.
The latter condition ensures that even by subselection of the operation output one cannot generate coherence from incoherent states. 
Every Kraus operator then must be of the form
\begin{equation}\label{GC:eq:incohKraus}
\hat{F}=\sum_i c_{i} \ket{i}\bra{\psi_i},
\end{equation}
in which $c_{i}\in\mathbb{C}$ and $\ket{\psi_i}\in{\rm span}\{\ket{j}\in\E_i\}$ so that $\E_i$s are disjoint subsets of $\E$.

The second class of operations are called {\it strict} incoherent operations.
They are simply incoherent operations of the above given form with the extra restriction that for every Kraus operator $\hat{F}_i$ it holds true that $\hat{F}_i^\dag$ is also incoherent so that the adjoint map of $\Lambda$ given by $\Lambda^\ddagger(\cdot)=\sum_i \hat{F}_i^\dag(\cdot)\hat{F}_i$ is also incoherent~\cite{Winter2016}.
Using Eq.~\eqref{GC:eq:incohKraus} this implies that the Kraus operators of strict incoherent operations has the form
$\hat{F}=\sum_i c_{i} \ket{i}\bra{j(i)}$, with both $\ket{i},\ket{j}\in\E$ and $j(i)$ being a one-to-one function. 
As shown by Yadin~{\it et al}~\cite{Yadin2016}, these operations correspond to those also not consuming quantum coherence within the given computational basis.
A necessary and sufficient for strictness of incoherent operations is given below.
\begin{lemma}\label{GC:lem:SincCond}
\cite{Meznaric2013,Yadin2016,Chitambar2016-2}
An incoherent operation $\Lambda$ is strict incoherent if and only if it possesses a set of Kraus operators $\{\hat{F}_i\}$ such that for all quantum states $\hat{\varrho}$ holds
\begin{equation}\label{GC:eq:SincCond}
\forall i:\quad \Delta[\hat{F}_i\hat{\varrho}\hat{F}_i^\dag] = \hat{F}_i\Delta[\hat{\varrho}]\hat{F}_i^\dag.
\end{equation}
\end{lemma}
\noindent Here, $\Delta[\cdot]=\sum_i \bra{i}\cdot\ket{i}\ketbra{i}$ is the fully depolarizing map.
Equation~\eqref{GC:eq:SincCond} is sometimes notationally compressed into a commutation relation as $[\Delta,\hat{F}_i] = 0$, where the implicit multiplication must be understood as a concatenation of superoperators.

In nonlocal scenarios, the relevant class of operations to our study is local incoherent operations and classical communication (LICC) where Alice and Bob perform only incoherent operations and share their possible outcomes via a classical channel.
We also note that $\text{LICC}\subset\text{S}$~\cite{Streltsov2017,Chitambar2016}.

\section{Nonclassicality in Quantum Coherence and Quantum Computation}


\subsection{A Computational Perspective on Nonclassicality}\label{Sec:CompNonCl}

Let us begin with the definition of ``a nonclassical physical process'' from a computational perspective, highlighting the role of computational efficiency in our physical picture.
A valid empirical theory is one which is plausibly testable and falsifiable as per below.
Let $\T$ be a physical theory and $\PP$ a physical process consisting of preparations, transformations, and measurements of some physical system.
Then, testing the theory $\T$ in process $\PP$ consists of three steps:
\begin{enumerate}
\item Write down the equations provided by $\T$ that are assumed to govern $\PP$.
\item {\it Efficiently compute} the predictions of $\T$ regarding the outcomes of measurements on the outputs of $\PP$.
This can be deterministic or probabilistic, e.g.,\ quantum theory is intrinsically nondeterministic.
Here, by efficient we mean probabilistic in polynomial time.
\item Compare the predictions of $\T$ against the experimental results obtained in $\PP$ and check the validity of the theory.
\end{enumerate}
\noindent First note that, we only speak of the properties of {\it processes} rather than {\it systems} with respect to given theories.
Second, for the above procedure to be consistent it is crucial that the computation used in step 2 be itself efficiently described by $\T$.
To clarify the reason, denote the specific computational process leading to predictions of $\T$ for $\PP$ by $\PP^\star$.
Both $\PP^\star$ and $\PP$ are physical processes irrespective of the presumed underlying theory, therefore, the fact that the particular computation $\PP^\star$ resembles $\PP$ also means that $\PP$ replicates $\PP^\star$.
It immediately follows that if a theory other than $\T$, say $\T^\star$, is necessary for efficiently describing $\PP^\star$, then it must be necessary for an efficient description of $\PP$ too.
Equivalently, if we assume that $\T$ provides a sufficient explanation for $\PP$, then it must also recount $\PP^\star$.  
All theories including classical ones, quantum mechanics, general relativity, etc., are physical theories that have been subject to such tests.
A non-$\T$-process can now defined as follows.
\begin{definition}\label{GC:def:Tproc}
A process $\PP$ is said not to be a $\T$-process (or, said to be a non-$\T$-process) if and only if $\T$ fails the three-step validity test in $\PP$.
\end{definition}
\noindent For instance, spectroscopic measurement of a black-body radiation is a nonclassical process because a classical theory fails to give an account for it in terms of the above test.

The important point that is commonly missed in assigning the adjective ``$\T$'' to a process, however, is the role of the computational efficiency in step 2.
Suppose that we have a theory $\T^\star$ for which we cannot efficiently compute (at least approximately up to some error $\varepsilon$) the result of its equations for a given physical process $\PP^\star$ on a computer efficiently described by $\T^\star$. 
Then, it would be practically implausible for us to figure out if $\T^\star$ passes the validity test in $\PP^\star$.
In other words, we do not have a way to determine within a reasonable time if $\T^\star$ is the suitable theory for describing $\PP^\star$ without running into contradictions.
Therefore, it is meaningless to consider $\PP^\star$ a $\T^\star$-process.
Similarly, if there exists a process for which we cannot efficiently compute the predictions of the classical theory on a classical computer, it cannot carry the prefix ``classical''.
Hence, in information science and from an operational perspective, all classical physical processes are premised to be \emph{efficiently} simulatable on a (probabilistic or deterministic) classical Turing machine, corresponding to the $\rm BPP$ class of computational complexity.
\begin{crit}\label{GC:crit:classicality}
A physical process that cannot be efficiently simulated on a classical computer is nonclassical.
\end{crit}
\noindent We emphasize here that, Criterion~\ref{GC:crit:classicality} only provides a sufficient condition, meaning that, not every nonclassical process is not efficiently simulatable on classical computers.
For instance, many quantum processes can be efficiently classically simulated.


\subsection{Nonclassicality in Resource Theory of Coherence}

Superpositions of generic states of a physical system are not allowed in classical theories.
Hence, coherence is considered to be a unique feature of post-classical theories.
Thus, one can also investigate nonclassicality within the framework of quantum theory of coherence.
It is sometimes stated that strict incoherent operations are classical ones in the resource theory of coherence because they are represented by stochastic transformations with respect to the computational basis $\E$, resembling a classical process (see e.g.,\ Refs.~\cite{Meznaric2013,Yadin2016}).
This conclusion, however, is debatable in view of our discussion in Sec.~\ref{Sec:CompNonCl} as follows.
Consider the map $\Upsilon\s{D}=\Delta\circ\Upsilon^\star\circ\Delta$ where $\Upsilon^\star$ is a quantum computation, where we assume that it is fixed by a given set of input parameters up to the input quantum state.
$\Delta[\cdot]{=}\sum \langle i|\cdot|i\rangle |i\rangle\langle i|$ represents the fully depolarizing map with respect to the computational basis $\mathcal{E}{=}\{|i\rangle\}$ and $\circ$ is the composition operation between superoperators.
Now, assuming any Kraus representation of the computation process, say $\Upsilon^\star(\cdot)=\sum_i\hat{F}_i(\cdot)\hat{F}_i^\dag$, we have
\begin{equation}
\Upsilon\s{D}=\Delta\circ\Upsilon^\star\circ\Delta =\sum_{ijk}\ketbra{k}\hat{F}_i\ketbra{j}\cdot\ketbra{j}\hat{F}_i^\dag\ketbra{k}.
\end{equation}
Due to the fact that the set of input states of a fixed size to the computation are finite, the Hilbert space, and consequently the space of operators acting on it, are considered to be finite dimensional.
Hence, we can define new indices $r=(i,j,k)$ so that $\Upsilon\s{D} =\sum_{r}c_rc_r^*\ket{s(r)}\bra{r}\cdot\ket{r}\bra{s(r)}$, with $\ket{s(r)}\in\set\s{inc}$ and $c_r=\bra{k}\hat{F}_i\ket{j}$.
This gives the set of Kraus operators for the map $\Upsilon\s{D}$ as $\{\hat{G}_r=c_r\ket{s(r)}\bra{r}\}$.
It can readily be seen that, from Lemma~\ref{GC:lem:SincCond}, $\Upsilon\s{D}$ is strict incoherent.
Notice that, because inputs and outputs of a quantum computer can always be considered to be computational basis states (i.e., incoherent states), the two depolarizing maps in $\Upsilon\s{D}$ leave them unchanged and do not affect the computation.
Therefore, for a classical user of the quantum computer, $\Upsilon\s{D}$ is computationally as powerful as $\Upsilon^\star$, implying that it cannot be classical even though it is strict incoherent.
The catch is that, there exist incoherent input states for which the map $\Upsilon\s{D}$ is not {\it efficiently} decomposable into a polynomial number of operations from a finite set of universal stochastic operations.
Such inputs correspond to the cases in which the map performs a quantum computation.

The important class of incoherent operations to our discussion in this section is thus a subset of strict incoherent operations that we name {\it universal strict incoherent} (USI) denoted by $\OO\s{USI}$. 
Elements of $\OO\s{USI}$ are those generating the symmetric group (i.e., the group of permutations) on $\E$.
In mathematical terms ${\rm sym}{\E}=\langle\!\langle \OO\s{USI} \rangle\!\rangle$, where $\langle\!\langle \cdot \rangle\!\rangle$ is the group generation operation via group composition, that is, the group is formed by repeatedly composing the elements of the generating set.
It is also straightforward to show that ${\rm sym}\E$ is isomorphic to ${\rm sym}\{1,\dots,d\}$, where $d$ is the dimensionality of $\E$.
The universality of $\OO\s{USI}$ must be understood over the set of incoherent states, that is, USI operations are necessary and sufficient to transform any incoherent state $\hat{\sigma}\in\set\s{inc}$ to any other incoherent state $\hat{\sigma}'\in\set\s{inc}$ via their composition and mixing, and subselection of outcomes.
Note also that, starting from a pure state the subselection can be disregarded.
Important to this construction is that any strict incoherent operation can be obtained from a (not necessarily efficient) composition of the elements in $\OO\s{USI}$. 
That is,
\begin{equation}\label{GC:eq:IncEffDec}
\forall \Gamma\in\OO\s{inc}: \quad \exists \{\Lambda_i\}\subseteq\OO\s{USI},
\end{equation}
such that $\Gamma[\cdot] = \Lambda_1\circ\Lambda_2\circ\Lambda_3\circ\cdots[\cdot]$.
The final remark is that, $\Upsilon\s{D}$ does not belong to $\OO\s{USI}$. 
This is because, $\Upsilon\s{D}\in\OO\s{USI}$ may hold only if we know the output of the process for all pure incoherent inputs.
However, given the fact that all the other input computation parameters are fixed, the latter implies that $\Upsilon\s{D}$, and thus $\Upsilon^\star$, would not be a computation since we already know the output of the process for all relevant input states.

Speaking of the separation between classical and quantum correlations we should first be clear about what we mean by ``nonclassicality''.
We thus first propose the following notion of ``classicality'' within the context of coherence theory.
\begin{definition}\label{GC:def:classicalObs}
Within the context of coherence theory, a classical observer is one who is restricted to universal strict incoherent operations $\OO\s{USI}$, their probabilistic mixture, and subselection.
\end{definition}
\noindent Here, the universality of strict incoherent operations implies that every \emph{incoherent} state can be obtained via a successive operation of $\OO\s{USI}$ elements on another incoherent state, their convex combination and subselection of the outcomes; see Fig.~\ref{GC:fig:Classicality}.

\begin{figure}[t!]
\begin{center}
  \includegraphics[width=0.8\columnwidth]{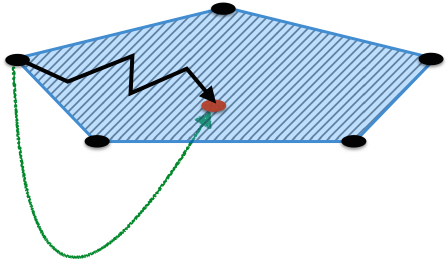}
\end{center}
  \vspace{-0em}
  \caption[Geometrical illustration of nonclassicality within the context of coherence theory]{\textbf{Geometrical illustration of nonclassicality within the context of coherence theory.}
  A classical process (the zigzag solid black line) is represented as an efficient composition of strictly incoherent operations evolving inside the set of incoherent states at all times.
  A quantum process from a classical observer's point of view (the green curve), on the other hand, is equivalent to a strict incoherent operation.
  However, it may involve generation and consumption of coherence at some stages and thus, it may partially be traversing outside the incoherent set.
   Such maps may or may not be efficiently representable as a composition of USI operations.
  }
  \label{GC:fig:Classicality}
\end{figure}

Definition~\ref{GC:def:classicalObs} naturally gives rise to the following sufficient condition for nonclassicality within the context of coherence theory.
\begin{crit}\label{GC:crit:CohNoncl}
For a classical observer equipped with a set of USI operations $\OO\s{USI}$, a process that cannot be efficiently represented as a compositions of $\OO\s{USI}$ elements and their convex combinations for at least one input state is nonclassical.
Here, by efficiency we mean a polynomial number of USI maps in the size of the input state.
\end{crit}
\noindent We also emphasize here that, not every nonclassical process is not efficiently decomposable into universal strictly incoherent operations.


\subsection{The Equivalence Theorem}

We now state the first result of the present manuscript which establishes a fundamental link between classical computation and quantum coherence formalism.
\begin{theorem}\label{CompCohIso}
There exists an isomorphism between classical computations and the formalism of coherence theory equipped with a set of USI operations.
\end{theorem}
\begin{proof}
The claimed isomorphism can be constructed as follows.
\begin{enumerate}
\item The set of all possible states of a deterministic classical computer (bits) can be represented as elements of a finite, but sufficiently large, set $\set\s{p}=\{s_i\}_{i\in\I}$ with the index set $\I={1,\dots,N}$ for some $N<\infty$.
They are perfectly distinguishable and thus, they can be mapped onto an orthonormal basis set of vectors within a Hilbert space as $\M\s{ps}:\set\s{p}\to\{\ketbra{i}\}_{i\in\I}$.
These vectors form the computational basis $\E=\{\ketbra{i}\}_{i\in\I}$.
In a probabilistic classical computer, the input as well as the readout state of the computation could be a probabilistic mixture of the pure state elements as $s=\sum_i p_is_i$ for $s_i\in\set\s{p}$ such that $\vs{p}=(p_1,\dots,p_N)$ is a vector of probabilities with $\sum_i p_i=1$.
Hence, the state space of such a computer is $\set\s{cl}=\overline{\conv\set\s{p}}$.
Clearly there is a bijection between elements of $\set\s{cl}$ and $\set\s{inc}$ as $\M\s{ms} : \set\s{cl} \to \set\s{inc}$ with $\M\s{ms} [s]=\sum_i p_i \ketbra{i}=\hat{\sigma}$.

The converse is also true.
Given a computational basis $\E=\{\ketbra{i}\}_{i\in\I}$ with the index set $\I={1,\dots,N}$ for some $N<\infty$, one can define the map $\M\s{ps}^{-1}:\E\to\set\s{p}$, where $\set\s{p}=\{s_i\}_{i\in\I}$ is a set of distinguishable states identifying different preparations of pure inputs to a classical computer.
Similarly, given a mixed incoherent state $\hat{\sigma}\in\set\s{inc}$, one can define a vector of probabilities $\vs{p}=(p_1,\dots,p_N)$ for which $\hat{\sigma}=\sum_i p_i\ketbra{i}$ and then map it onto a probabilistic state of a classical computer via $\M\s{ms}^{-1}:\set\s{inc}\to\set\s{cl}$ where $\set\s{cl}=\overline{\conv\set\s{p}}$ and $\M\s{ms}^{-1}[\hat{\sigma}] = \sum_i p_i s_i$.

\item Every classical algorithm running on a classical computer can be decomposed into a sequence of successive operations of universal classical logic gates from a finite set $\G\s{UCL}$.
Each logic gate is represented by a stochastic map acting on the state $s$ of the computer~\cite{NielsenBook}.
Importantly, such gates do not create or consume superpositions of computational states and thus, are represented by strictly incoherent transformations with respect to the defined computational basis.
As a result, the class of universal classical gates is mapped onto a subset of strictly incoherent operations as $\M\s{ops}:\G\s{UCL}\to\OO\s{USI}$ where $\M\s{ops}$ is bijective and can be implicitly defined as follows.
For every classical gate $G \in\G\s{UCL}$ and every incoherent state $\hat{\sigma}\in\set\s{inc}$, $\hat{\sigma}'=\Lambda[\hat{\sigma}]=\M\s{ops}[G][\hat{\sigma}] = \M\s{ms}\circ G \circ\M\s{ms}^{-1}[\hat{\sigma}]$.
The invertibility of $\M\s{ms}$ simply implies the invertibility of $\M\s{ops}$: for every USI quantum gate $\Lambda \in\OO\s{USI}$ and every computational state $s\in\set\s{cl}$, $s'=G{s}=\M\s{ops}^{-1}[\Lambda][s] = \M\s{ms}^{-1}\circ \Lambda \circ\M\s{ms}[s]$. 
The universality of the classical logic gates then immediately implies the universality of the strictly incoherent maps $\Lambda$ defined above over the set of incoherent states.
\end{enumerate}
The two steps above can be summarized as
\begin{equation}\label{GC:eq:CompCohIso}
\begin{split}
\set\s{cl}\quad &\xleftrightarrow[~\M\s{ms}^{-1}~]{\M\s{ms}}\quad \set\s{inc},\\
\G\s{UCL}\quad &\xleftrightarrow[~\M\s{ops}^{-1}~]{\M\s{ops}}\quad \OO\s{USI},
\end{split}
\end{equation}
establishing an isomorphism between classical computation and the structure of incoherent states equipped with a USI set of operations.
\end{proof}
Now, we use the fact that any efficient classical computation can be implemented in a polynomial number of steps in combination with the above isomorphism to conclude that the efficiency of classical algorithms implies application of a polynomial number of USI operations.
The converse is also obvious.
Consequently, every physical process that cannot be efficiently simulated on a classical computer, cannot be represented as an efficient composition of USI maps within some coherence theory and vice versa.
We formalize this in a theorem highlighting the connection between the two nonclassicality criteria above as our second result.
\begin{theorem}\label{GC:th:CohNonConnection}
The nonclassicality Criteria~\ref{GC:crit:classicality} and~\ref{GC:crit:CohNoncl} are equivalent.
\end{theorem}

Our third result, which is an immediate consequence of Theorems~\ref{CompCohIso} and~\ref{GC:th:CohNonConnection}, and proved within Appendix~\ref{App:2}, tells us when it is \emph{not} possible to do quantum computation.
\begin{theorem}\label{GC:th:QConNecQComp}
Production or consumption of quantum coherence provides the necessary resource for the exponential speed~up of quantum computations versus classical ones.
\end{theorem}
\noindent Having these results at hand, we can now answer the question ``What do we learn about quantum correlations from collaborative quantum computing?''


\section{A Toy Protocol}

Our aim is to use Criterion~\ref{GC:crit:classicality}, i.e., the power of quantum computation models, to show the quantumness of correlations.
To this end, we need to consider quantum computation protocols that are nonlocal and use mixed quantum states,
noting that any pure correlated quantum state is necessarily entangled.
We thus first consider a toy protocol for which we already have available the minimal tools for a characterization of the resources used.

In classical computations, to run large computational tasks on multiple supercomputers in parallel and then combine their outputs to get a final result is a common protocol; a model called \emph{distributed computing}.
It is thus intriguing to consider a situation in which each of the servers is equipped with a quantum computer to run quantum computations.
In such scenarios, the input to each server is possibly classical information accompanied with quantum states.
We consider two versions of a distributed computing task in which a client, Charlie, exploits nondiscordant states to run quantum computations on two servers, Alice and Bob.
We assume that the following rules apply:
(i)~servers are forbidden to communicate;
(ii)~they do not have access to any sources of quantum states ---they can only perform unitary transformations and make destructive measurements on their outputs;
(iii)~the client, on the other hand, does not possess any quantum processors ---he may only have limited capability of preparing quantum states.
We also assume that there are no losses, inefficiencies, or errors, as they are not essential to our arguments and conclusions about the quantumness of the correlations.

\paragraph*{Task 1.---}
Charlie has classical descriptions of two $n_{\rm X}$-qubit unitary matrices $\hat{U}_{\rm X}$ (${\rm X}{=}{\rm A},{\rm B}$) in terms of polynomial sized network of universal gates.
His task is to estimate the quantity 
\begin{equation}\label{iota}
\iota = \frac{{\rm Tr}\hat{U}_{\rm A}{\cdot}{\rm Tr}\hat{U}_{\rm B}}{2^{n_{\rm A}+n_{\rm B}}}.
\end{equation}
He can also prepare up to two pure qubit states---any other states are maximally mixed.
 
There is strong evidence to suggest that estimating the normalised trace of a $n$-qubit unitary matrix $\mathrm{Tr}\hat{U}/2^n$, generated from a polynomial sized network of universal gates, is hard for a classical computer~\cite{Datta2005,Datta2007,Datta2008,Morimae2014,Fujii2014,Fujii2016}. 
These hardness arguments imply that even with a classical description of the polynomial sized network forming the unitary, Charlie (as well as Alice and Bob) cannot efficiently estimate the normalised trace $\iota$ using his classical resources.
The latter follows from the fact that, assuming an exponential growth in the classical resources required to estimate the normalized trace of one unitary ${\rm Tr}\hat{U}_{\rm X}/2^{n_{\rm X}}$ ($\rm X=A,B$) with the number of input qubits $n_{\rm X}$, classically estimating their normalized tensor product will also require at least an exponential effort in the total number of input qubits $n_{\rm A}+n_{\rm B}$.
Thus, he is encountering a classically challenging task.
He can, however, conquer the difficulty with the help of the two servers using a duplicated DQC1 protocol~\cite{DQC1}, termed here as \textit{nonlocal deterministic quantum computing with two qubits} (NDQC2).

\begin{figure}[h]
  \includegraphics[width=\columnwidth]{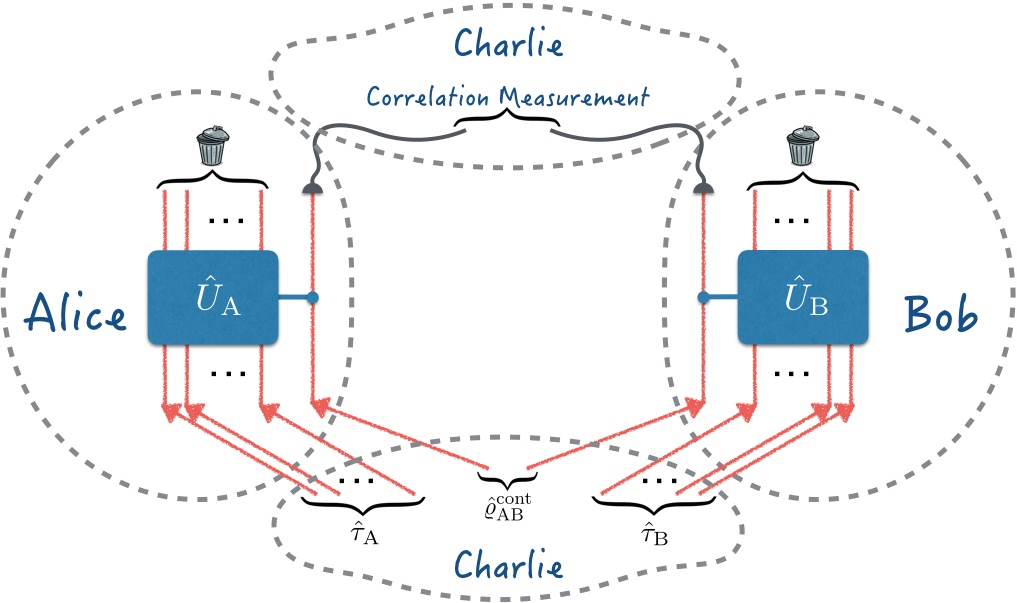}
  \caption{\textbf{The schematic of a nonlocal deterministic quantum computation with two qubits (NDQC2).}
  A client, Charlie, who does not possess any quantum processors aims to estimate the quantity $\iota ={\rm Tr}\hat{U}_{\rm A}{\cdot}{\rm Tr}\hat{U}_{\rm B}/2^{n_{\rm A}+n_{\rm B}}$.
  This is believed to be hard to perform on a classical computer.
  Therefore, he asks two servers, Alice and Bob, who are capable of performing unitary transformations and making destructive measurements to realize the controlled unitaries $\hat{U}_{\rm X}^{\rm cont}$ (${\rm X}{=}{\rm A},{\rm B}$).
  The servers are forbidden to communicate and do not have access to any sources of quantum states.
  Charlie then sends  strictly classical states to the servers and receives the results of the measurements of the Pauli operators as per Eq.~\eqref{CorrM}.
  By manipulating the received data, he is able to efficiently estimate $\iota$.
  Depending on his choice of state, Charlie is also able to hide the local estimates from Alice and Bob without reducing his global computational power. 
  }\label{NDQC2Scheme}
\end{figure}

\paragraph*{NDQC2 (see Fig.~\ref{NDQC2Scheme}).---}
We assume it is always possible for Charlie to ask Alice and Bob to realize the controlled unitaries $\hat{U}_{\rm X}^{\rm cont}{=} |0\rangle_{\rm X}\langle 0|{\otimes} \hat{I}_{\rm X} {+} |1\rangle_{\rm X}\langle 1|{\otimes} \hat{U}_{\rm X}$ (${\rm X}{=}{\rm A},{\rm B}$), respectively~\cite{Barenco1995}.
He then prepares the two-qubit control system in either of the two pure product (nondiscordant) states $\hat{\varrho}^{\rm cont}_{{\rm AB};1}{=}|\pm\rangle_{\rm A}\langle \pm|\otimes|\pm\rangle_{\rm B}\langle \pm|$, where $|\pm\rangle_{\rm X}{=}(|0\rangle_{\rm X} \pm |1\rangle_{\rm X})/\sqrt{2}$, and the ancillary qubits in the maximally mixed states, $\hat{\tau}_{\rm X}{=}\hat{\mathbbmss{I}}^{\otimes n_{\rm X}}/2^{n_{\rm X}}$, and sends them to the servers.
Alice and Bob operate on their respective ancillae and control inputs, locally and independently, to obtain $\hat{\varrho}^{\rm out}_{{\rm X};1}{=}\hat{U}_{\rm X}^{\rm cont} (|\pm\rangle_{\rm X}\langle \pm|{\otimes}\hat{\tau}_{\rm X}) \hat{U}_{\rm X}^{{\rm cont}\dag}$, and make measurements of the Pauli operators on the output control states,
${\rm Tr} \hat{\varrho}^{\rm out}_{{\rm X};1} [\left(\hat{\bm{\sigma}}_{x;{\rm X}} {+} \bm{i}\hat{\bm{\sigma}}_{y;{\rm X}}\right){\otimes}\hat{\mathbbmss{I}}^{\otimes n_{\rm X}}] {=} {\pm} {\rm Tr}\hat{U}_{\rm X}/2^{n_{\rm X}}$.
Finally, the servers send their statistics to Charlie, who will combine them to obtain an estimate of
\begin{equation}\label{CorrM}
\iota{=}{\prod_{{\rm X}={\rm A},{\rm B}}}{\left\langle\hat{\bm{\sigma}}_{x;{\rm X}} {+} \bm{i}\hat{\bm{\sigma}}_{y;{\rm X}}  \right\rangle} {=} \left\langle{\bigotimes_{{\rm X}={\rm A},{\rm B}}}\left(\hat{\bm{\sigma}}_{x;{\rm X}} {+} \bm{i}\hat{\bm{\sigma}}_{y;{\rm X}}\right) \right\rangle,
\end{equation}
where $\langle\cdot\rangle$ denotes the quantum expectation value of the output control state.

\paragraph*{Task 2.---}
Consider Task 1 where Charlie also wants to hide the local estimates ${\rm Tr}\hat{U}_{\rm X}/2^{n_{\rm X}}$ from Alice and Bob at all times, given the constraints (i)-(iii) on the protocol.

Clearly Task 2 is classically, if not impossible, as hard as Task 1, because estimating the normalized trace of the global unitary $\hat{U}_{\rm A}\otimes\hat{U}_{\rm B}/2^{n_{\rm A}+n_{\rm B}}$ is also classically hard.
However, this can be done efficiently using NDQC2 protocol above if Charlie prepares the control in the nondiscordant state
\begin{equation}\label{ABinput}
\hat{\varrho}_{{\rm AB};2}^{\rm cont}=\frac{1}{2}\sum_{x=\pm}|x\rangle_{\rm A}\langle x|{\otimes}|x\rangle_{\rm B}\langle x|.
\end{equation}
Following the same procedure as in Task~1 and making the same measurements, the correlations within the measurement outcomes are processed by Charlie as per the r.h.s of Eq.~\eqref{CorrM}, which results in $\iota {=} \left\langle\left(\hat{\bm{\sigma}}_{x;{\rm A}} {+} \bm{i}\hat{\bm{\sigma}}_{y;{\rm A}}\right) \left( \hat{\bm{\sigma}}_{x;{\rm B}} {+} \bm{i}\hat{\bm{\sigma}}_{y;{\rm B}}\right)  \right\rangle$.
The marginals of the control state~\eqref{ABinput} are maximally mixed states, so that independent measurements do not result in any information about $\iota$.
This hides the local estimates from Alice and Bob.
Hence, NDQC2 enables a classically hard collaborative task only using correlated inputs and correlation measurements.


\subsection{Quantum Correlations in NDQC2.}

As we discussed earlier, any process that cannot be efficiently simulated on a classical computer is \emph{nonclassical}.
We have shown that NDQC2 cannot be efficiently simulated using classical resources, i.e., classical communication \emph{and} classical computers held by Alice and Bob, implying that it is an example of a nonlocal nonclassical process.

We are now able to give an operational meaning to the term ``quantum correlations'', inspired by the NDQC2 toy model.
First, from Criterion~\ref{GC:crit:classicality}, we infer the quantumness of the resources used in collaborative quantum computations from the nonclassical advantages obtained in them.
Second, whenever the locally accessible quantum resources are inadequate to fully account for such advantages, they are necessarily the result of correlations.
A closer look at the NDQC2 protocol shows that in Tasks~2, not only the input, but also the output state is nondiscordant with respect to the Alice-Bob partitioning in which the correlations are measured, because the global operation $\hat{U}_{\rm AB}^{\rm cont}{=}\hat{U}_{\rm A}^{\rm cont}{\otimes}\hat{U}_{\rm B}^{\rm cont}$ preserves local orthonormality of bases.
In addition, from Eq.~\eqref{ABinput}, it is clear that there is no local entanglement or discord within each server in Task~2 in contrast to the DQC1 protocol~\cite{Datta2005,Datta2008}.
In fact, in this case, Alice and Bob have no local quantum computational resources as they only receive locally maximally-mixed states.
Hence, from the common perspective of quantum information theory, the input and output states in this scenario are considered to possess no quantum correlations between Alice and Bob.
One thus should wonder if there is nothing quantum going on locally, and there is nothing quantum about the correlations between Alice and Bob as characterized by the standard measures of quantum information, then where does the quantum power of the {\it joint} Alice-Bob party in NDQC2 come from?
And why is that obtained only through correlation measurements?
Our answer is that the computational power of NDQC2 in this case is indeed a manifestation of quantumness of correlations distributed between Alice and Bob through the input quantum state.
Importantly, the standard classification of quantum correlations does not account for these sort of correlations.

One might object to calling these correlations quantum by considering the following scenarios: 
(i) suppose that the local servers in our toy model are granted the ability to prepare quantum states.
Then, Charlie can send classical encrypted messages instructing the servers to prepare either the superposition state $|+\rangle_{\rm X}$ or $|-\rangle_{\rm X}$ with equal probabilities to Alice and Bob, where they would have created the state $\hat{\varrho}_{{\rm AB};2}^{\rm cont}$ in Eq.~\eqref{ABinput} locally without accessing the content of the message, and thus, simulating the protocol locally.
Regardless of the complexity of such a \emph{semi-classical} protocol compared to ours, we emphasize that the possibility to prepare the states via LOCC does not imply the classicality of its inherent correlations.
Similarly, any separable discordant state can be prepared using LOCC, and yet it is believed that quantum discord implies quantum correlations;
(ii) if Alice and Bob have access to local quantum resources, i.e., perfect qubits, they may extract the correlations encoded within the input state and access the local estimates of ${\rm Tr}\hat{U}_{\rm X}/2^{n_{\rm X}}$.
Equivalently, if they are allowed to communicate during the protocol, then they can obtain the same information from the correlations as Charlie does.
In this case, the resolution is that the extractibility of the encoded probability distribution and the final result of the computation also does not imply the classicality of the correlations within quantum states.
The counterexample is, for instance, a one-way discordant state.
If the two-way classical communication is allowed between parties, then they can extract the probability distribution encoded within such states.
However, one-way discordant states are also quantum correlated.
The lesson we learn is thus that the quantumness of correlations can be revealed only if \emph{appropriate restrictions} are imposed on particular tasks.
In our case, the required restrictions are exactly (i)-(iii) given for the NDQC2.




\subsection{Global Coherence in NDQC2}
In order to characterize quantum correlations present in NDQC2, we should identify the resources empowering it.
We observed in Theorem~\ref{GC:th:QConNecQComp} that quantum coherence is necessary for the exponential speed-up of quantum computers.
Recently, it has been shown that coherence also provides the sufficient resource for the particular case of DQC1 protocol in the sense that the precision of the quantity estimated in DQC1 is a function of the amount of quantum coherence inherent within the input state~\cite{Matera2016}.
Since NDQC2 enjoys a construction similar to DQC1, we anticipate that the power of NDQC2 in Tasks~1 and~2 is also due to the coherence of the input and output states.
Here we show this fact quantitatively.

We start by choosing the relative entropy of coherence (REC)~\cite{Baumgratz2014,Winter2016} as our measure of coherence.
For any density operator $\hat{\varrho}$, the REC is given by
\begin{equation}
\mathfrak{C}_{\rm r}(\hat{\varrho}){:=}\mathtt{S}(\Delta[\hat{\varrho}]){-}\mathtt{S}(\hat{\varrho}),
\end{equation}
in which $\mathtt{S}(\hat{\varrho}){=}{-}{\rm Tr}\hat{\varrho}\log{\hat{\varrho}}$ is the von Neumann entropy.
We omit the dependence of $\mathfrak{C}_{\rm r}$ is the basis from the function's argument for brevity.
REC satisfies the three main requirements for any faithful measure of quantum coherence:
(i) a quantum state is incoherent if and only if $\mathfrak{C}_{\rm r}(\hat{\varrho}){=}0$;
(ii) it is nonincreasing on average under all incoherent transformations, i.e., $\mathfrak{C}_{\rm r}(\hat{\varrho})\geqslant \sum_i p_i\mathfrak{C}_{\rm r}(\hat{\varrho_i})$ where $p_i$ is the probability of obtaining $\hat{\varrho}_i$ upon measurement;
(iii) it is convex, i.e., $\mathfrak{C}_{\rm r}(\sum_i p_i\hat{\varrho}_i)\leqslant \sum_i p_i\mathfrak{C}_{\rm r}(\hat{\varrho_i})$ for any set of states $\{\hat{\varrho}_i\}$ and probability distribution $\{p_i\}$.
Operationally, REC is equivalent to the distillable coherence and quantifies the optimal rate at which maximally coherent states can be prepared from infinitely many copies of a given mixed state using incoherent operations~\cite{Winter2016}.

First, we choose the local computational bases in our protocol to be $\mathcal{E}_{\rm X}{=}\{|0\rangle_{\rm X}\otimes|\xi_i\rangle_{\rm X},|1\rangle_{\rm X}\otimes|\xi_i\rangle_{\rm X}\}_{i=1}^{n_{\rm X}}$, where $\{|\xi_i\rangle_{\rm X}\}_{i=1}^{n_{\rm X}}$ are eigenvectors of $\hat{U}_{\rm X}$ for ${\rm X}={\rm A},{\rm B}$.
Now, we see that, in both Tasks~1 and~2, the input states $\hat{\varrho}^{\rm in}_{{\rm AB};1}{=}\hat{\varrho}^{\rm cont}_{{\rm AB};1}{\otimes}\hat{\tau}_{\rm A}{\otimes}\hat{\tau}_{\rm B}$ and $\hat{\varrho}^{\rm in}_{{\rm AB};2}{=}\hat{\varrho}^{\rm cont}_{{\rm AB};2}{\otimes}\hat{\tau}_{\rm A}{\otimes}\hat{\tau}_{\rm B}$ are globally coherent with respect to the global computational basis $\mathcal{E}_{\rm AB}{=}\mathcal{E}_{\rm A}{\otimes}\mathcal{E}_{\rm B}$ with $\mathfrak{C}_{\rm r}(\hat{\varrho}^{\rm in}_{{\rm AB};1}) {=} 2\log{2}$ and $\mathfrak{C}_{\rm r}(\hat{\varrho}_{{\rm AB};2}^{\rm in}){=}\log{2}$, respectively.

Second, we consider the local coherences of the marginal states in Tasks~1 and~2 to obtain $\mathfrak{C}_{\rm r}(|\pm\rangle_{\rm X}\langle \pm|{\otimes}\hat{\tau}_{\rm X}) {=} \log{2}$ and $\mathfrak{C}_{\rm r}(\hat{\mathbbmss{I}}_{\rm X}/2{\otimes}\hat{\tau}_{\rm X}) {=} 0$, respectively in each task.
From Refs.~\cite{Ma2016,Matera2016} we know that the less the input coherence to a DQC1 protocol is, the worse the estimation of the normalized trace of the unitary will be.
Therefore, these values justify the fact that in Task~1 local traces are accessible to Alice and Bob, due to the \emph{local} computational powers provided by \emph{locally} coherent resources, while in Task~2 they remain hidden to them because no \emph{local} computational power is available to parties.
We emphasize here that $\hat{U}_{\rm AB}^{\rm cont}{=}\hat{U}_{\rm A}^{\rm cont}{\otimes}\hat{U}_{\rm B}^{\rm cont}$ neither increases nor decreases the amount of global and local REC, as $\mathcal{E}_{\rm AB}$ is an eigenbasis of $\hat{U}_{\rm AB}^{\rm cont}$~\cite{Peng2016}.

Now, we show that only global coherence plays a role in the nonclassical performance of NDQC2 protocol. 
\begin{lemma}\label{NDQC2Prec}
The precision of the estimated quantity $\iota$ in Eq.~\eqref{iota} is given by the amount of global coherence inherent within the input quantum state as quantified by REC.
\end{lemma}
\noindent Please see Appendix~\ref{App:1} for the proof. 
We now make use the fact that local coherence of marginal states implies the global coherence of the joint state of the system, in combination with Theorem~\ref{GC:th:QConNecQComp} and the above lemma, to conclude the fourth main result of the present paper.
\begin{theorem}\label{NDQC2resource}
Global coherence is necessary and sufficient for the nonclassical performance of the NDQC2 protocol.
\end{theorem}
\noindent Theorem~\ref{NDQC2resource}, clearly shows the role of global coherence in our protocol.
There is, however, a difference between NDQC2 of Task~1 and Task~2 to be discussed shortly.


\section{Net Global Coherence as Quantum Correlations}

In the Task~2, we conclude the quantumness of the correlations from nonclassical performance of the protocol, since, (i) the input state to the protocol is indeed correlated; our considerations are merely regarding whether they are quantum or classical, and,
(ii) other than observing correlations between Alice and Bob outcomes, Charlie would not be able to obtain the result of the task.
In Task~1, on the other hand, such a conclusion is not valid.
The obvious reason is that, the product state used in Task~1 represents independent preparation procedures and hence, by postulates of quantum mechanics, uncorrelated states.

To address this difference, suppose that a measure of coherence $\mathfrak{C}$ has been chosen to characterize quantum correlations in some computational basis.
We then define
\begin{equation}\label{Cnet1}
\mathfrak{C}^{\rm net}(\hat{\varrho}_{\rm AB}) = \mathfrak{C} (\hat{\varrho}_{\rm AB}) - \mathfrak{C}(\hat{\varrho}_{\rm A}) - \mathfrak{C}(\hat{\varrho}_{\rm B}),
\end{equation}
to determine the {\it net-global} quantum computational-power of the quantum states with the interpretation that we subtract the local quantum powers from the overall one.
The requirement that product states show no quantum correlations, and hence no global computational power except those due to local resources, imposes the condition ``if $\hat{\varrho}_{\rm AB}=\hat{\varrho}_{\rm A}{\otimes}\hat{\varrho}_{\rm B}$ then $\mathfrak{C}^{\rm net}(\hat{\varrho}_{\rm AB}) {=} 0$''.
This holds true if and only if the coherence measure $\mathfrak{C}$ is additive, i.e., $\mathfrak{C} (\hat{\varrho}_{\rm A}{\otimes}\hat{\varrho}_{\rm B}){=}\mathfrak{C} (\hat{\varrho}_{\rm A}){+}\mathfrak{C}(\hat{\varrho}_{\rm B})$.
Importantly, the relative entropy of coherence~\cite{Baumgratz2014,Winter2016} is an additive measure, while, for instance, the $\ell_1$-norm of coherence~\cite{Baumgratz2014} is not.
It immediately follows that in Task~1 $\mathfrak{C}_{\rm r}^{\rm net}(\hat{\varrho}_{{\rm AB};1}^{\rm in}){=}0$, that is, all the global quantum computational power is due to the local resources.
In sharp contrast, in Task~2 $\mathfrak{C}_{\rm r}^{\rm net}(\hat{\varrho}_{{\rm AB};2}^{\rm in}){=} \log{2}$, interpreted as the amount of global quantum computational power purely due to quantum correlations.
We thus notice that in Task~2 all the computational advantage can be associated with the net global coherence.
We draw inspiration from this fact and, with a little foresight, define the quantum correlated states as per below.

\begin{definition}
(Quantum-Correlated States)
A bipartite quantum state $\hat{\varrho}_{\rm AB}$ is said to contain quantum correlations with respect to a global computational basis $\mathcal{E}_{\rm AB}$ if and only if $\mathfrak{C}_{\rm r}^{\rm net}(\hat{\varrho}_{\rm AB}){>}0$ within $\mathcal{E}_{\rm AB}$.
\end{definition}
\noindent In what follows, we show that this definition is indeed well justified, demonstrating that previously known classes of quantum correlations are emergent from our extended notion, and that it allows for a proper operational interpretation.


\subsection{Some Properties of $\mathfrak{C}_{\rm r}^{\rm net}$}

Using relative entropy of coherence, as shown in Appendix~\ref{DEQ4}, we can rewrite Eq.~\eqref{Cnet1} as
\begin{equation}\label{Cnet2}
\mathfrak{C}_{\rm r}^{\rm net}(\hat{\varrho}_{\rm AB})=\mathtt{I}(\hat{\varrho}_{\rm AB})-\mathtt{I}(\Delta_{\rm AB}[\hat{\varrho}_{\rm AB}]),
\end{equation}
in which $\mathtt{I}(\hat{\varrho}_{\rm AB}){:=}\mathtt{S}(\hat{\varrho}_{\rm A}){+}\mathtt{S}(\hat{\varrho}_{\rm B}){-}\mathtt{S}(\hat{\varrho}_{\rm AB})$ is the mutual information and $\Delta_{\rm AB}$ is the bipartite fully dephasing channel within the global computational basis.
$\mathfrak{C}_{\rm r}^{\rm net}$ is a further generalization of the \emph{basis-dependent} discord in which only one of the parties undergoes the dephasing~\cite{Yadin2016}.
According to Definition~\ref{GC:def:classicalObs} of a classical observer in coherence theory, the perception of two classical observers from the quantum world is limited to the globally incoherent state $\hat{\sigma}\s{AB}=\Delta_{\rm AB}[\hat{\varrho}_{\rm AB}]$.
Hence, the quantity $\mathtt{I}(\Delta_{\rm AB}[\hat{\varrho}_{\rm AB}])$ can be considered as the mutual information between these two classical observers with particular local bases $\E\s{A}$ and $\E\s{B}$.
It then follows that $\mathfrak{C}_{\rm r}^{\rm net}(\hat{\varrho}_{\rm AB})$ represents the net quantum information shared between the two in the global basis $\mathcal{E}_{\rm AB}$.
It is also worth pointing out that a quantity called {\it global quantum discord} was previously introduced by Rulli and Sarandy in Ref.~\cite{Rulli2011}. 
The important difference between this quantity and net global coherence, however, is that a minimization over all computational bases is involved in the definition of the former, in tradition of discord quantities.
Here, in contrast, we showed, within the framework of coherence theory, that the computational basis plays a significant role in the characterisation of quantum correlations.
In this regard, we obtain the following properties of the net global coherence.

\begin{theorem}\label{th:positNetCoh}
For every bipartite quantum state $\hat{\varrho}_{\rm AB}$ it holds that $\mathfrak{C}^{\rm net}_{\rm r}(\hat{\varrho}_{\rm AB})\geqslant 0$.
The equality holds if and only if the quantum state is a product state or has the form $\hat{\varrho}_{AB}=\sum_{ij} p_{ij} \ketbrax{i}{A}\otimes\ketbrax{j}{B}$ with respect to the global computational basis $\mathcal{E}_{\rm AB }{=}\mathcal{E}_{\rm A}{\otimes}\mathcal{E}_{\rm B}{=}\{|i\rangle_{\rm A}\otimes|j\rangle_{\rm B}\}$, with $\{p_{ij}\}$ being a probability distribution.
\end{theorem}

\noindent Please see Appendix~\ref{P1} for a detailed proof.
It is necessary that every extension of the standard classes of quantum correlated states includes the hierarchy of 
entangled and discordant states as special cases.
To show that indeed this holds true for our approach in a well-defined way, we first state a feature of $\mathfrak{C}^{\rm net}_{\rm r}$, the proof of which is given in Appendix~\ref{App:5}.

\begin{theorem}\label{th:NetCohPure}
Any multipartite pure state has a nonzero net global coherence if and only if it is entangled.
\end{theorem}

Second, we give some more general results on the global-coherence properties of the standard classification.
As it has been shown in Appendix~\ref{P3}, the following holds.

\begin{theorem}\label{th:nonDiscord}
A bipartite quantum state $\hat{\varrho}_{\rm AB}$ is nondiscordant if and only if $\mathfrak{C}_{\rm r}^{\rm net}(\hat{\varrho}_{\rm AB}){=}0$ within some appropriate global computational basis $\mathcal{E}^\star_{\rm AB}$.
\end{theorem}
 
\noindent The following is then an immediate conclusion.

\begin{corol}\label{BasInDepNetGlobCoh}
\cite{Girolami2013}
For every quantum state that is characterized as quantum correlated in quantum information theory (entangled and discordant states), $\mathfrak{C}_{\rm r}^{\rm net}(\hat{\varrho}_{\rm AB}){>}0$ independent of the chosen (pure product) global-basis.
\end{corol}

\noindent From Corollary~\ref{BasInDepNetGlobCoh} we see that the standard picture of quantum correlations follows from our formalism if a net global coherence is required to exist with respect to every computational basis.
The latter is indeed a very strong condition.
In particular, because classical observers are restricted to a specific computational basis, having global coherence with respect to that basis is sufficient for obtaining possible quantum advantage of the correlations, provided that appropriate fine-grained operations at a quantum level are accessible.

Arguably, one of the desirable properties of a theory of quantum correlations as a resource, which is not met by quantum discord, is convexity~\cite{Matera2016}.
Important to our construction is that there exist convex combinations of two product states that possess quantum correlations with a positive $\mathfrak{C}_{{\rm r}}^{\rm net}(\hat{\varrho}_{\rm AB})$.
This shows that convex combinations of uncorrelated states are \emph{not} necessarily uncorrelated.
In other words, the two desirable properties of ``a convex resource theory of quantum correlations'' and ``the preparation independence of the product states'' seem incompatible, unless considering a theory of nonclassicality which identifies products of coherent states as nonclassical~\cite{Sperling2018}.

As a final word, quantum-correlated states provide an interesting nonlocal feature which we call \emph{coherence localization}., namely, the process of providing local computational power for one party with the aid of another.
The trivial case is that there already exists some local coherence available to parties. 
However, there exists a nontrivial scenario that is given below and proved in Appendix~\ref{App:7}.

\begin{theorem}
\cite{Streltsov2017}
Given the local computational bases $\mathcal{E}_{\rm A}{=}\{|i\rangle_{\rm A}\}$ and $\mathcal{E}_{\rm B}{=}\{|j\rangle_{\rm B}\}$, and Alice and Bob sharing a bipartite quantum state $\hat{\varrho}_{\rm AB}$, they can distil quantum coherence on Bob's side using LICC if and only if $\hat{\varrho}_{\rm AB}$ cannot be written as $\hat{\varrho}_{\rm AB}{=}\sum_{j} p_{j} \hat{\varrho}_{{\rm A};j}{\otimes}|j\rangle_{\rm B}\langle j|$.
\end{theorem}

\begin{corol}\label{CoherenceInduce}
Given the local computational bases $\mathcal{E}_{\rm A}{=}\{|i\rangle_{\rm A}\}$ and $\mathcal{E}_{\rm B}{=}\{|j\rangle_{\rm B}\}$, Alice and Bob sharing a bipartite quantum state $\hat{\varrho}_{\rm AB}$, and $\mathfrak{C}_{\rm r}(\hat{\varrho}_{\rm A}){=}\mathfrak{C}_{\rm r}(\hat{\varrho}_{\rm B}){=}0$, they cannot distil quantum coherence on neither sides using LICC if and only if $\hat{\varrho}_{\rm AB}$ is not quantum correlated, i.e. $\mathfrak{C}_{\rm r}^{\rm net}(\hat{\varrho}_{\rm AB}){=}0$, with respect to the global bases $\mathcal{E}_{\rm AB}{=}\mathcal{E}_{\rm A}{\otimes}\mathcal{E}_{\rm B}$.
\end{corol}

\noindent The above corollary gives us a way to see quantum correlations from a different perspective.
Suppose that a bipartite state is shared between two classical agents Alice and Bob.
If we ask under what conditions one of the parties can provide quantum coherence (and hence, computational power) for the other party using local classical operations and classical communication, i.e., LICC, then the answer is given by Corollary~\ref{CoherenceInduce}:
if $\mathfrak{C}_{\rm r}^{\rm net}(\hat{\varrho}_{\rm AB}){\neq}0$ with respect to the global bases $\mathcal{E}_{\rm AB}{=}\mathcal{E}_{\rm A}{\otimes}\mathcal{E}_{\rm B}$.
This is simply the operational interpretation of our extended notion of quantum correlations.


\section*{Discussion}

We have built a rigorous framework for benchmarking the efficiency of computational models as a separation criterion of classical and quantum resources.
We then introduced the NDQC2 model of quantum computation, a collaborative nonlocal algorithm for estimating the product of the normalized traces of two unitary matrices that shows an exponential speedup compared to the best known classical algorithms.
We demonstrated that this task can be done using a separable and nondiscordant input state. 
In essence, based on our quantum-classical separation criterion we argued that the exponential speedup of NDQC2 over classical algorithms is a manifestation of the quantum correlations beyond entanglement and discord inherent within the input quantum state to our toy model.

Our main observation, however, was that the standard classification of quantum correlations in quantum information theory does not capture the quantumness of such correlations, and thus require a revision.
It is noteworthy that, similar arguments exists within the quantum optics community, where nonclassical phase-space quasiprobability distributions are the signatures of quantumness~\cite{VogelBook,Ferraro2012}.
This viewpoint has also been extended to composite discrete-continuous variables systems~\cite{Sperling2017}.
Very recently, we have shown that phase-space nonclassicality provides a resource for nonlocal \textsc{BosonSampling} in absence of the standard quantum correlations of quantum information~\cite{Shahandeh2017NLBS}, in favour of the quantum optical viewpoint.

The approach we presented here, extends the standard quantum information theoretic classification of quantum correlations to include quantum advantages obtained in distributed quantum computation protocols.
We quantitatively showed that the net global coherence emerging from correlations between subsystems can be considered as equivalent to quantum correlations.
We showed that our generalized definition of quantum correlations characterizes the necessary and sufficient quantum resources in NDQC2 and properly contains the standard classification as a particular case.

It is worth noting that within the rapidly developing field of quantum coherence, there has recently been interests in establishing an appropriate framework for the paradigm of local incoherent operations and classical communication (LICC), in which parties are restricted to locally incoherent operations~\cite{Streltsov2017,Chitambar2016}.
One can think of LICC as the class of classical operations on distant multipartite systems.
Quantum correlations, as we have defined, then are the weakest in the sense that they allow a party to remotely provide quantum computational resources for a distant party using LICC.
Thus, we see that it is possible to obtain local resource states for efficient quantum computation, even if no local resource states are initially available and the parties are locally restricted to classical operations which do not generate resource states, if and only if they share quantum correlations of the type introduced here.

The relation between multipartite quantum coherence and quantum correlations has also been studied recently by other researchers~\cite{Killoran2016,Ma2016,Streltsov2015,Tan2016}.
However, their approaches are fundamentally different from the perspective presented in this manuscript.
Specifically, rather than to investigate the conversion of \emph{local} quantum coherence into standard types of quantum correlations, we have considered \emph{net global} quantum coherence as a primitive notion of quantum correlation, which can be distributed and might reveal its unique quantum signatures only within distributed quantum protocols.
As a consequence, our results opens up the possibility for exploring new protocols that use such correlations, which are generically cheaper than entanglement and discord, to perform collaborative tasks more efficiently than any classical algorithm.


\section*{Acknowledgements}The authors gratefully acknowledge Werner Vogel, Fabio Costa, and Eric Chitambar for valuable discussions.
This project was supported by the Australian Research Council Centre of Excellence for Quantum Computation and Communication Technology (CE110001027).

\newpage
\appendix
\begin{widetext}

\section{Proof of Theorem~\ref{GC:th:QConNecQComp}}\label{App:2}

\noindent\textbf{Theorem~\ref{GC:th:QConNecQComp}.}
Production or consumption of quantum coherence provides the necessary resource for the exponential speed~up of quantum computations versus classical ones.

\begin{proof}
First, notice that a quantum computer runs algorithms in a polynomial time (number of steps or gates), otherwise it would not be efficient.
Then, the contrary of the above statement would be that,
\begin{itemize}
\item  there exist a quantum computer in which all steps are strict incoherent operations, i.e., neither they produce nor they consume quantum coherence.
\end{itemize}
We now use the fact that one is able to produce any incoherent state $\hat{\sigma}\in\set\s{inc}$ at the output of a quantum computer, which is equivalent to the universality of the computer's gates over the set of incoherent states as given by Eq.~\eqref{GC:eq:IncEffDec}.
Hence, for one such a computer, the set of gates used is, in fact, a USI set of operations.
As a result, our counter assumption further translates into that, 
\begin{itemize}
\item there exists a quantum computer which merely performs a polynomial number of USI operations.
\end{itemize}
Using the isomorphism between classical computers and the formalism of coherence theory with a set of USI operations given in Eq.~\eqref{GC:eq:CompCohIso}, the latter just means that whatever such a quantum computer does can be equally performed on a classical computer.
This contradicts the assumption that we have a quantum computer running algorithms faster than any classical computer, hence the result.
\end{proof}


\section{Proof of Lemma~\ref{NDQC2Prec}: Estimate Precision in NDQC2}\label{App:1}

Following Ref.~\cite{Matera2016}, it is easy to show that the precision in the estimate of the quantity $\iota {=}{\rm Tr}\hat{U}_{\rm A}{\cdot}{\rm Tr}\hat{U}_{\rm B}/2^{n_{\rm A}+n_{\rm B}}$ is given by the relative entropy of coherence of the control qubit $\mathfrak{C}_{\rm r}(\hat{\varrho}_{{\rm AB}}^{\rm cont})$.
The latter is defined as the optimal rate at which maximally coherent states can be distilled from infinitely many copies of the state $\hat{\varrho}_{{\rm AB}}^{\rm cont}$ using incoherent operations~\cite{Winter2016}.
Hence, given $M$ copies of the state $\hat{\varrho}_{{\rm AB}}^{\rm cont}$ a total number $N{\approx}M\mathfrak{C}_{\rm r}(\hat{\varrho}_{{\rm AB}}^{\rm cont})$ pairs of maximally coherent states can be distilled in Alice's and Bob's local laboratories.
Assuming $M$ being large, as shown in~~\cite{Matera2016}, the standard error in the estimation of each local normalized trace $\iota_{\rm X} {=}{\rm Tr}\hat{U}_{\rm X}/2^{n_{\rm X}}$ is given by $\textsc{SE}(\iota_{\rm X})\approx\sqrt{\frac{2-|\iota_{\rm X}|^2}{N}}=\sqrt{\frac{2-|\iota_{\rm X}|^2}{M\mathfrak{C}_{\rm r}(\hat{\varrho}_{{\rm AB}}^{\rm cont})}}$ for $\rm X= A,B$.
The standard error of the quantity $\iota$ is thus given by
\begin{equation}\label{SE}
\textsc{SE}(\iota)=\sqrt{\textsc{SE}(\iota_{\rm A})^2+\textsc{SE}(\iota_{\rm B})^2}\approx\sqrt{\frac{4-|\iota_{\rm A}|^2-|\iota_{\rm B}|^2}{M\mathfrak{C}_{\rm r}(\hat{\varrho}_{{\rm AB}}^{\rm cont})}}.
\end{equation}

The binary precision in the estimation of a number $\mu$ with $|\mu|{\leqslant} 1$ goes like ${\rm BP}(\mu){\approx} -\log_2[{\rm SE}(\mu)]$.
Noting that the nominator in Eq.~\eqref{SE} is bounded by $2\leqslant 4 -|\iota_{\rm A}|^2-|\iota_{\rm B}|^2 \leqslant 4$, we find
\begin{equation}
{\rm BP}(\iota)\approx \frac{1}{2} \log_2 \mathfrak{C}_{\rm r}(\hat{\varrho}_{{\rm AB}}^{\rm cont}).
\end{equation}

\section{Derivation of Eq.~\eqref{Cnet2} of the main text} \label{DEQ4}

To show the truth of the Eq.~\eqref{Cnet2} of the main text, again, we recall from Eq.~\eqref{Cnet1} of the main text that
\begin{equation}\label{Crnet}
\mathfrak{C}_{\rm r}^{\rm net}(\hat{\varrho}_{\rm AB}) = \mathfrak{C}_{\rm r} (\hat{\varrho}_{\rm AB}) - \mathfrak{C}_{\rm r}(\hat{\varrho}_{\rm A}) - \mathfrak{C}_{\rm r}(\hat{\varrho}_{\rm B}),
\end{equation}
in which $ \mathfrak{C}_{\rm r}$ refers to the REC.
We also restate that $\mathfrak{C}_{\rm r}(\hat{\varrho}){:=}\mathtt{S}(\Delta[\hat{\varrho}]){-}\mathtt{S}(\hat{\varrho})$, in which $\mathtt{S}(\hat{\varrho}){=}{-}{\rm Tr}\hat{\varrho}\log{\hat{\varrho}}$ and $\Delta[\hat{\varrho}]{:=}\sum \langle i|\hat{\varrho}|i\rangle |i\rangle\langle i|$ are the von Neumann entropy and the fully dephasing channel with respect to the computational bases $\mathcal{E}{=}\{|i\rangle\}$.
Consequently, we have
\begin{equation}\label{Crs}
\begin{split}
&\mathfrak{C}_{\rm r}(\hat{\varrho}_{\rm AB}){=}\mathtt{S}(\Delta_{\rm AB}[\hat{\varrho}_{\rm AB}]){-}\mathtt{S}(\hat{\varrho}_{\rm AB}),\\
&\mathfrak{C}_{\rm r}(\hat{\varrho}_{\rm A}){=}\mathtt{S}(\Delta_{\rm A}[\hat{\varrho}_{\rm A}]){-}\mathtt{S}(\hat{\varrho}_{\rm A}),\\
&\mathfrak{C}_{\rm r}(\hat{\varrho}_{\rm B}){=}\mathtt{S}(\Delta_{\rm B}[\hat{\varrho}_{\rm B}]){-}\mathtt{S}(\hat{\varrho}_{\rm B}).
\end{split}
\end{equation}
Substituting the relations of Eq.~\eqref{Crs} into Eq.~\eqref{Crnet}, after a simple rearrangement and using $\mathtt{I}(\hat{\varrho}_{\rm AB}){:=}\mathtt{S}(\hat{\varrho}_{\rm A}){+}\mathtt{S}(\hat{\varrho}_{\rm B}){-}\mathtt{S}(\hat{\varrho}_{\rm AB})$, we get
\begin{equation}\label{CrIDI}
\begin{split}
\mathfrak{C}_{\rm r}^{\rm net}(\hat{\varrho}_{\rm AB}) &= \left\{ \mathtt{S}(\hat{\varrho}_{\rm A}) + \mathtt{S}(\hat{\varrho}_{\rm B}) - \mathtt{S}(\hat{\varrho}_{\rm AB}) \right\} - \left\{ \mathtt{S}(\Delta_{\rm A}[\hat{\varrho}_{\rm A}]) + \mathtt{S}(\Delta_{\rm B}[\hat{\varrho}_{\rm B}]) - \mathtt{S}(\Delta_{\rm AB}[\hat{\varrho}_{\rm AB}]) \right\},\\
&=\mathtt{I}(\hat{\varrho}_{\rm AB})-\mathtt{I}(\Delta_{\rm AB}[\hat{\varrho}_{\rm AB}]).
\end{split}
\end{equation}
We note that, for the second equality to be true, we need that ${\rm Tr}_{\rm B}\Delta_{\rm AB}[\hat{\varrho}_{\rm AB}] = \Delta_{\rm A}[{\rm Tr}_{\rm B}\hat{\varrho}_{\rm AB}] = \Delta_{\rm A}[\hat{\varrho}_{\rm A}]$, and similarly ${\rm Tr}_{\rm A}\Delta_{\rm AB}[\hat{\varrho}_{\rm AB}] = \Delta_{\rm B}[{\rm Tr}_{\rm A}\hat{\varrho}_{\rm AB}] = \Delta_{\rm B}[\hat{\varrho}_{\rm B}]$.
This can be easily verified, as
\begin{equation}
\begin{split}
{\rm Tr}_{\rm B}\Delta_{\rm AB}[\hat{\varrho}_{\rm AB}] &= \sum_{ij}\langle i,j|\hat{\varrho}_{\rm AB}|i,j\rangle |i\rangle\langle i|\left({\rm Tr}|j\rangle\langle j|\right) \\
&= \sum_i |i\rangle\langle i|\otimes \langle i|\left(\sum_j \langle j|\hat{\varrho}_{\rm AB}|j\rangle\right)|i\rangle \\
&= \sum_i |i\rangle\langle i|\otimes \langle i|\left({\rm Tr}_{\rm B}\hat{\varrho}_{\rm AB}\right)|i\rangle \\
&=\Delta_{\rm A}[\hat{\varrho}_{\rm A}].
\end{split}
\end{equation}

\qed

\section{Proof of Theorem~\ref{th:positNetCoh}}\label{P1}

\noindent\textbf{Theorem~\ref{th:positNetCoh}.}
For every bipartite quantum state $\hat{\varrho}_{\rm AB}$ it holds that $\mathfrak{C}^{\rm net}_{\rm r}(\hat{\varrho}_{\rm AB})\geqslant 0$.
The equality holds if and only if the quantum state is a product state or has the form $\hat{\varrho}_{AB}=\sum_{ij} p_{ij} \ketbrax{i}{A}\otimes\ketbrax{j}{B}$ with respect to the global computational basis $\mathcal{E}_{\rm AB }{=}\mathcal{E}_{\rm A}{\otimes}\mathcal{E}_{\rm B}{=}\{|i\rangle_{\rm A}\otimes|j\rangle_{\rm B}\}$, with $\{p_{ij}\}$ being a probability distribution.

\paragraph*{Proof.}
Recall that the basis-dependent discord~\cite{Yadin2016} is defined as 
\begin{equation}\label{DAB}
\mathtt{D}_{{\rm A}\to{\rm B}}(\hat{\varrho}_{\rm AB}):=\mathtt{I}(\hat{\varrho}_{\rm AB}) - \mathtt{I}(\Delta_{\rm A}[\hat{\varrho}_{\rm AB}]),
\end{equation}
in which $\Delta_{\rm A}[\hat{\varrho}_{\rm AB}]=\sum_i |i\rangle\langle i|\otimes\langle i|\hat{\varrho}_{\rm AB}|i\rangle = \sum_i p_i |i\rangle\langle i|\otimes \hat{\varrho}_{{\rm B};i}$, with $p_i={\rm Tr}_{\rm B}\langle i|\hat{\varrho}_{\rm AB}|i\rangle$, is the one-sided dephasing channel for Alice's subsystem. 
It has been shown that $\mathtt{D}_{{\rm A}\to{\rm B}}(\hat{\varrho}_{\rm AB})\geqslant 0$~\cite{Ollivier2001}.
Similarly,
\begin{equation}\label{Rev}
\mathtt{D}_{{\rm B}\to{\rm A}}(\hat{\varrho}_{\rm AB})\geqslant 0 \quad\Leftrightarrow\quad \mathtt{I}(\hat{\varrho}_{\rm AB}) \geqslant \mathtt{I}(\Delta_{\rm B}[\hat{\varrho}_{\rm AB}]).
\end{equation}
We can easily verify that
\begin{equation}
\begin{split}
\Delta_{\rm AB}[\hat{\varrho}_{\rm AB}] &= \sum_{ij}\langle i,j|\hat{\varrho}_{\rm AB}|i,j\rangle |i\rangle\langle i|\otimes|j\rangle\langle j| \\
&= \sum_i |i\rangle\langle i|\otimes \langle i|\left(\sum_j |j\rangle\langle j|\otimes\langle j|\hat{\varrho}_{\rm AB}|j\rangle\right)|i\rangle \\
&=\Delta_{\rm A}[\Delta_{\rm B}[\hat{\varrho}_{\rm AB}]] = \Delta_{\rm B}[\Delta_{\rm A}[\hat{\varrho}_{\rm AB}]] .
\end{split}
\end{equation}
Combining this with Eq.~\eqref{DAB} and \eqref{Rev}, we have
\begin{equation}\label{CPos}
\begin{split}
\mathtt{D}_{{\rm A}\to{\rm B}}(\Delta_{\rm B}[\hat{\varrho}_{\rm AB}]) &=\mathtt{I}(\Delta_{\rm B}[\hat{\varrho}_{\rm AB}]) - \mathtt{I}(\Delta_{\rm A}[\Delta_{\rm B}[\hat{\varrho}_{\rm AB}]])\\
&= \mathtt{I}(\Delta_{\rm B}[\hat{\varrho}_{\rm AB}]) - \mathtt{I}(\Delta_{\rm AB}[\hat{\varrho}_{\rm AB}])\geqslant 0\\
& \Leftrightarrow \mathtt{I}(\Delta_{\rm B}[\hat{\varrho}_{\rm AB}]) \geqslant \mathtt{I}(\Delta_{\rm AB}[\hat{\varrho}_{\rm AB}])\\
& \Rightarrow \mathtt{I}(\hat{\varrho}_{\rm AB}) \geqslant \mathtt{I}(\Delta_{\rm B}[\hat{\varrho}_{\rm AB}]) \geqslant \mathtt{I}(\Delta_{\rm AB}[\hat{\varrho}_{\rm AB}])\\
& \Rightarrow \mathfrak{C}_{\rm r}^{\rm net}(\hat{\varrho}_{\rm AB})=\mathtt{I}(\hat{\varrho}_{\rm AB})-\mathtt{I}(\Delta_{\rm AB}[\hat{\varrho}_{\rm AB}]) \geqslant 0.
\end{split}
\end{equation}

For the equality to hold, either both $\mathtt{I}(\hat{\varrho}_{\rm AB})$ and $\mathtt{I}(\Delta_{\rm AB}[\hat{\varrho}_{\rm AB}])$ must be zero, which implies that the state is a product.
Or, one must have $\mathtt{D}_{{\rm A}\to{\rm B}}(\Delta_{\rm B}[\hat{\varrho}_{\rm AB}]) = 0$ in Eq.~\eqref{CPos}, which implies the form of $\Delta_{\rm B}[\hat{\varrho}_{\rm AB}]$ to be $\Delta_{\rm B}[\hat{\varrho}_{\rm AB}]= \sum_{ij} p_{ij} \ketbrax{i}{A}\otimes\ketbrax{j}{B}$.
This in turn means that $\hat{\varrho}_{AB}=\sum_{ij} p_{ij} \ketbrax{i}{A}\otimes\hat{\varrho}_{{\rm B};j}$. 
Also, from symmetry of A and B, and by similar arguments, it must be true that $\hat{\varrho}_{AB}=\sum_{ij} p_{ij} \hat{\varrho}_{{\rm A};i}\otimes\ketbrax{j}{B}$.
Together, we must have $\hat{\varrho}_{AB}=\sum_{ij} p_{ij} \ketbrax{i}{A}\otimes\ketbrax{j}{B}$.

\qed

\section{Proof of Theorem~\ref{th:NetCohPure}}\label{App:5}
\noindent\textbf{Theorem~\ref{th:NetCohPure}.}
Any multipartite pure state has a nonzero net global coherence if and only if it is entangled.

\paragraph*{Proof.}
Using Theorem~\ref{th:positNetCoh} it is clear that any entangled state has a nonzero net global coherence.
The converse easily follows from the fact that any pure multipartite state is either a product state or entangled so that the net global coherence of nonentangled states becomes zero by the additivity condition for the coherence measures.

\qed

\section{Proof of Theorem~\ref{th:nonDiscord}}\label{P3}

\noindent\textbf{Theorem~\ref{th:nonDiscord}.}
A bipartite quantum state $\hat{\varrho}_{\rm AB}$ is a CC state if and only if $\mathfrak{C}_{\rm r}^{\rm net}(\hat{\varrho}_{\rm AB}){=}0$ within some appropriate global computational basis $\mathcal{E}^\star_{\rm AB}$.

To prove Theorem~\ref{th:nonDiscord}, let us first prove the following useful Lemma.

\noindent\textbf{Lemma.} A bipartite quantum state $\hat{\varrho}_{\rm AB}$ is a CC state if and only if it is incoherent within some appropriate global computational basis $\mathcal{E}^\star_{\rm AB}$.

\paragraph*{Proof.} 

If: Assuming that there exists a computational basis $\mathcal{E}^\star_{\rm AB}{=}\{|i\rangle_{\rm A}{\otimes}|j\rangle_{\rm B}\}$ in which the state $\hat{\varrho}_{\rm AB}$ is incoherent implies that the state can be written as $\hat{\varrho}_{\rm AB}{=}\sum_{ij} p_{ij} |i\rangle_{\rm A}\langle i|{\otimes}|j\rangle_{\rm B}\langle j|$. 
Due to the orthogonality of the local vectors in the basis set, $\hat{\varrho}_{\rm AB}$ is CC.
Only if: A CC state, by definition, admits the form $\hat{\varrho}_{\rm AB}{=}\sum_{ij} p_{ij} |i\rangle_{\rm A}\langle i|{\otimes}|j\rangle_{\rm B}\langle j|$, which is clearly incoherent with the choice of computational basis $\mathcal{E}^\star_{\rm AB}{=}\{|i\rangle_{\rm A}{\otimes}|j\rangle_{\rm B}\}$.

\qed

Now we can prove Theorem~\ref{th:nonDiscord}.

\paragraph*{Proof.} 

Using Lemma~1 above, it is sufficient to prove that a bipartite quantum state $\hat{\varrho}_{\rm AB}$ is incoherent within some  appropriate global computational basis $\mathcal{E}^\star_{\rm AB}$ if and only if $\mathfrak{C}_{\rm r}^{\rm net}(\hat{\varrho}_{\rm AB}){=}0$.

Only if: A CC state, by definition, admits the form $\hat{\varrho}_{\rm AB}{=}\sum_{ij} p_{ij} |i\rangle_{\rm A}\langle i|{\otimes}|j\rangle_{\rm B}\langle j|$, which is clearly incoherent with the choice of the computational basis $\mathcal{E}^\star_{\rm AB}{=}\{|i\rangle_{\rm A}{\otimes}|j\rangle_{\rm B}\}$.
Also, global incoherence implies marginal incoherence, evidently.
Consequently, one has $\mathfrak{C}_{\rm r}(\hat{\varrho}_{\rm AB}){=}\mathfrak{C}_{\rm r}(\hat{\varrho}_{\rm A}){=}\mathfrak{C}_{\rm r}(\hat{\varrho}_{\rm B}){=}0$, which results $\mathfrak{C}_{\rm r}^{\rm net}(\hat{\varrho}_{\rm AB}){=}0$.

If: Assuming that $\mathfrak{C}_{\rm r}^{\rm net}(\hat{\varrho}_{\rm AB}){=}0$, from Eq.~\eqref{CrIDI} we have $\mathtt{I}(\hat{\varrho}_{\rm AB}) = \mathtt{I}(\Delta_{\rm AB}[\hat{\varrho}_{\rm AB}])$.
Combining this with the facts $\mathtt{I}(\hat{\varrho}_{\rm AB}) \geqslant \mathtt{I}(\Delta_{\rm A}[\hat{\varrho}_{\rm AB}]) \geqslant \mathtt{I}(\Delta_{\rm AB}[\hat{\varrho}_{\rm AB}])$ and $\mathtt{I}(\hat{\varrho}_{\rm AB}) \geqslant \mathtt{I}(\Delta_{\rm B}[\hat{\varrho}_{\rm AB}]) \geqslant \mathtt{I}(\Delta_{\rm AB}[\hat{\varrho}_{\rm AB}])$ (see Eq.~\eqref{CPos}), we conclude that $\mathtt{I}(\hat{\varrho}_{\rm AB}) = \mathtt{I}(\Delta_{\rm B}[\hat{\varrho}_{\rm AB}])$ and $\mathtt{I}(\hat{\varrho}_{\rm AB}) = \mathtt{I}(\Delta_{\rm B}[\hat{\varrho}_{\rm AB}])$.
On one hand, using Eq.~\eqref{DAB}, these equalities imply that the state must possesses a vanishing basis-dependent discord~\cite{Yadin2016} both from Alice to Bob and from Bob to Alice.
On the other hand, the standard discord is the minimum of the basis-dependent discord over all measurements on either sides~\cite{Yadin2016}.
As a result, the state must have a vanishing (standard) discord both from Alice to Bob and from Bob to Alice.
It is already known that this is possible only if the state is CC.

\qed

\section{Proof of Corollary~\ref{CoherenceInduce}}\label{App:7}

\noindent\textbf{Corollary~\ref{CoherenceInduce}.}
Given the local computational bases $\mathcal{E}_{\rm A}{=}\{|i\rangle_{\rm A}\}$ and $\mathcal{E}_{\rm B}{=}\{|j\rangle_{\rm B}\}$, Alice and Bob sharing a bipartite quantum state $\hat{\varrho}_{\rm AB}$, and $\mathfrak{C}_{\rm r}(\hat{\varrho}_{\rm A}){=}\mathfrak{C}_{\rm r}(\hat{\varrho}_{\rm B}){=}0$, they cannot distil quantum coherence on neither sides using LICC if and only if $\hat{\varrho}_{\rm AB}$ is not quantum correlated, i.e. $\mathfrak{C}_{\rm r}^{\rm net}(\hat{\varrho}_{\rm AB}){=}0$, with respect to the global bases $\mathcal{E}_{\rm AB}{=}\mathcal{E}_{\rm A}{\otimes}\mathcal{E}_{\rm B}$.

\paragraph*{Proof.}
The condition $\mathfrak{C}_{\rm r}(\hat{\varrho}_{\rm A}){=}\mathfrak{C}_{\rm r}(\hat{\varrho}_{\rm B}){=}0$ reduces $\mathfrak{C}_{\rm r}^{\rm net}(\hat{\varrho}_{\rm AB}){=}0$ to $\mathfrak{C}_{\rm r}(\hat{\varrho}_{\rm AB}){=}0$.
This reads as Alice and Bob cannot distil quantum coherence on neither sides using LICC if and only if $\hat{\varrho}_{\rm AB}$ is not globally coherent.
Using Theorem~4, this means that the state must be of the forms $\hat{\varrho}_{\rm AB}{=}\sum_{i} p_{i} |i\rangle_{\rm A}\langle i|{\otimes}\hat{\varrho}_{{\rm B};i}$ and $\hat{\varrho}_{\rm AB}{=}\sum_{j} p_{j} \hat{\varrho}_{{\rm A};j}{\otimes}|j\rangle_{\rm B}\langle j|$, simultaneously.
Equivalently, it must be of the form $\hat{\varrho}_{\rm AB}{=}\sum_{ij} p_{ij} |i\rangle_{\rm A}\langle i|{\otimes}|j\rangle_{\rm B}\langle j|$, that is, globally incoherent.
\qed

\end{widetext}

\bibliography{GCandQC.bib}
\bibliographystyle{apsrev4-1new.bst}

\end{document}